\documentclass[final,1p,times,authoryear]{elsarticle}
\usepackage{amsmath}
\usepackage{amsthm}
\usepackage{amssymb}
\usepackage{amsfonts}
\usepackage{dsfont}

\usepackage{float}
\usepackage[plainpages=false,pdfpagelabels=true,colorlinks=true,citecolor=blue,hypertexnames=false]{hyperref}
\usepackage[ruled,vlined]{algorithm2e}
\usepackage{mathtools}
\usepackage{extarrows}
\usepackage{url}
\usepackage{balance}
\usepackage{multirow}
\usepackage{color}
\usepackage{diagbox}
\usepackage{booktabs}
\usepackage{makecell}
\newcolumntype{V}{!{\vrule width 1pt}}

\newtheorem{theorem}{Theorem}
\newtheorem{lemma}[theorem]{Lemma}

\newtheorem{definition}[theorem]{Definition}

\newtheorem{example}[theorem]{Example}
\newtheorem{remark}[theorem]{Remark}
\newtheorem{problem}[theorem]{Problem}

\newcommand{\gr}{Gr\"{o}bner }

\def\z{{\bf z}}
\def\F{{\mathbf{F}}}
\def\G{{\mathbf{G}}}

\begin{document}

\begin{frontmatter}

\title{On Minor Left Prime Factorization Problem for Multivariate Polynomial Matrices}

\author[baic,smbh]{Dong Lu}
\ead{donglu@buaa.edu.cn}

\author[klmm,ucas]{Dingkang Wang\corref{cor1}}
\ead{dwang@mmrc.iss.ac.cn}

\author[klmm,ucas]{Fanghui Xiao}
\ead{xiaofanghui@amss.ac.cn}

\cortext[cor1]{Corresponding author}

\address[baic]{Beijing Advanced Innovation Center for Big Data and Brain Computing, Beihang University, Beijing 100191, China}

\address[smbh]{School of Mathematical Sciences, Beihang University, Beijing 100191, China}

\address[klmm]{KLMM, Academy of Mathematics and Systems Science, Chinese Academy of Sciences, Beijing 100190, China}

\address[ucas]{School of Mathematical Sciences, University of Chinese Academy of Sciences, Beijing 100049, China}

\begin{abstract}
 A new necessary and sufficient condition for the existence of minor left prime factorizations of multivariate polynomial matrices without full row rank is presented. The key idea is to establish a relationship between a matrix and its full row rank submatrix. Based on the new result, we propose an algorithm for factorizing matrices and have implemented it on the computer algebra system Maple. Two examples are given to illustrate the effectiveness of the algorithm, and experimental data shows that the algorithm is efficient.
\end{abstract}

\end{frontmatter}

\section{Introduction}

 Multivariate polynomial matrix factorization is one of the most important operations in multidimensional systems, signal processing, and other related areas \citep{Bose1982,Bose2003}. The factorization problems of multivariate polynomial matrices have been extensively investigated and numerous algorithms have been developed to compute factorizations of multivariate polynomial matrices. Since the factorization problems have been solved for univariate and bivariate polynomial matrices \citep{Morf1977New,Guiver1982Polynomial,Liu2013New}, we only consider the case where the number of variables is greater than or equal to three.

 Using three important concepts proposed by \cite{Youla1979Notes}, there have been many publications studying matrix factorizations. \cite{Lin1999Notes} first proposed the existence problem for zero prime factorizations of multivariate polynomial matrices. \cite{Charoenlarpnopparut1999Multidimensional} first used \gr bases of modules to compute zero prime matrix factorizations of multivariate polynomial matrices. After that, \cite{Lin2008ATutorial} introduced some applications of \gr bases in the broad field of signals and systems. \cite{Lin2001A} put forward the famous Lin-Bose conjecture which was solved by \cite{Pommaret2001Solving,Wang2004On}. \cite{Mingsheng2005On} focused on the existence problem for minor prime factorizations of multivariate polynomial matrices, and gave a necessary and sufficient condition. \cite{Mingsheng2007On} designed an algorithm to compute factor prime factorizations of a class of multivariate polynomial matrices.

 In linear algebra as well as multidimensional systems, the factorization problems of multivariate polynomial matrices without full row rank are important and deserve some attention \citep{Youla1979Notes,Lin1999Notes}. Up to now, few results have been achieved on factorizations of multivariate polynomial matrices without full row rank \citep{Lin2001A,Guan2018,Guan2019}. Therefore, this paper focuses on factorization problems of multivariate polynomial matrices without full row rank. Motivated by the views in \cite{Lin2001A}, we try to use local properties to study the existence for minor prime factorizations of multivariate polynomial matrices without full row rank.

 The rest of the paper is organized as follows. In Section \ref{sec_PP}, we introduce some basic concepts and present the problem that we are considering. We present in Section \ref{sec_MR} a new necessary and sufficient condition for the existence of minor left prime factorizations of multivariate polynomial matrices without full row rank. In Section \ref{sec_AE}, we construct an algorithm based on the new result, and use two examples to illustrate the effectiveness of the algorithm. A comparison with Guan's algorithm and experimental data are presented in Section \ref{sec_edata}. We end with some concluding remarks in Section \ref{sec_conclusions}.

\section{Preliminaries and Problem}\label{sec_PP}

 Let $n$ be the number of variables, and $\z$ be the $n$ variables $z_1,\ldots,z_n$, where $n\geq 3$. Let $k[\z]$ be the polynomial ring in $\z$ over $k$, where $k$ is an algebraically closed field. Let $k[\z]^{l\times m}$ denote the set of $l\times m$ matrices with entries in $k[\z]$, where $l\leq m$. Let $\F\in k[\z]^{l\times m}$, we use $d_i(\F)$ to denote the greatest common divisor of all the $i\times i$ minors of $\F$, and $I_i(\mathbf{F})$ to represent the ideal generated by all the $i\times i$ minors of $\mathbf{F}$, where $1\leq i \leq l$ and we stipulate that $I_0(\mathbf{F}) = k[\z]$.

 We first recall the most important concept in the paper.

 \begin{definition}
  Let $\mathbf{F}\in k[\z]^{l\times m}$ be of full row rank. Then $\F$ is said to be an minor left prime (MLP) matrix if all the $l\times l$ minors of $\F$ are relatively prime, that is, $d_l(\mathbf{F})$ is a nonzero constant.
 \end{definition}

 Let $\mathbf{F}\in k[\z]^{m\times l}$ with $m\geq l$, an MRP matrix can be similarly defined. We refer to \cite{Youla1979Notes} for more details about the concepts of zero left prime (ZLP) matrices and factor left prime (FLP) matrices.

 An MLP factorization of a multivariate polynomial matrix is formulated as follows.

 \begin{definition}\label{matrix_factorization}
  Let $\mathbf{F}\in k[\z]^{l\times m}$ with rank $r$, where $1\leq r \leq l$. $\mathbf{F}$ is said to admit an MLP factorization if $\mathbf{F}$ can be factorized as
  \begin{equation}\label{gerneral-matirx-factorization}
   \mathbf{F} = \mathbf{G}_0\mathbf{F}_0
  \end{equation}
  such that $\mathbf{G}_0\in k[\z]^{l\times r}$, and $\mathbf{F}_0\in k[\z]^{r\times m}$ is an MLP matrix.
 \end{definition}

 When Youla and Gnavi studied the structure of $n$-dimensional linear systems, they obtained the following MLP factorization lemma by using matrix theory.

 \begin{lemma}\label{Youla-minor}
  Let $\mathbf{A} = \begin{bmatrix} \mathbf{A}_{11} & \mathbf{A}_{12} \\ \mathbf{A}_{21} & \mathbf{A}_{22} \end{bmatrix} \in k[\z]^{l\times m}$ with rank $r$, where $\mathbf{A}_{11}\in k[\z]^{r\times r}$ with ${\rm det}(\mathbf{A}_{11}) \neq 0$, $\mathbf{A}_{12}\in k[\z]^{r\times (m-r)}$, $\mathbf{A}_{21}\in k[\z]^{(l-r)\times r}$, $\mathbf{A}_{22}\in k[\z]^{(l-r)\times (m-r)}$, and $1\leq r \leq l$. If $[\mathbf{A}_{11} ~ \mathbf{A}_{12}]$ is an MLP matrix, then $\mathbf{A}_{21}\mathbf{A}_{11}^{-1}$ is a multivariate polynomial matrix and $\mathbf{A}$ has an MLP factorization
  \begin{equation}\label{Youla-minor-equ}
    \mathbf{A} = \begin{bmatrix}\mathbf{I}_{r\times r} \\ \mathbf{A}_{21}\mathbf{A}_{11}^{-1} \end{bmatrix}\begin{bmatrix}\mathbf{A}_{11} & \mathbf{A}_{12} \end{bmatrix}.
  \end{equation}
 \end{lemma}

 In order to state conveniently the problem of this paper, we introduce the following concepts and conclusions.

 \begin{definition}[\cite{Matsumura1989C}]
  Let $\mathcal{K}$ be a submodule of $k[\z]^{1\times m}$, and $J$ be a nonzero ideal of $k[\z]$. We define
  \begin{equation*}
   \mathcal{K} : J = \{ \vec{u}\in k[\z]^{1\times m} \mid J \vec{u} \subseteq \mathcal{K} \},
  \end{equation*}
  where $J \vec{u}$ is the set $\{f\vec{u} \mid f\in J\}$.
 \end{definition}

 Obviously, $\mathcal{K} \subseteq \mathcal{K} : J$. Let $\{f_1,\ldots, f_s\}\subset k[\z]$ be a \gr basis of $J$, then
 \begin{equation}\label{quotient-module-1}
   \mathcal{K} : J = \mathcal{K} :\langle f_1,\ldots,f_s\rangle = (\mathcal{K} : f_1) \cap \cdots \cap (\mathcal{K} : f_s).
 \end{equation}
 Here, we write $\mathcal{K} :\langle f \rangle$ as $\mathcal{K} : f$ for any $f\in k[\z]$.

\begin{definition}[\cite{Eisenbud2013}]
 Let $\mathcal{K}$ be a finitely generated $k[\z]$-module, and $k[\z]^{1\times l} \xlongrightarrow{\phi} k[\z]^{1\times m} \rightarrow \mathcal{K} \rightarrow 0$ be a presentation of $\mathcal{K}$, where $\phi$ acts on the right on row vectors, i.e., $\phi(\vec{u}) = \vec{u}\cdot\mathbf{F}$ for $\vec{u}\in k[\z]^{1\times l}$ with $\mathbf{F}$ being a presentation matrix corresponding to the linear mapping $\phi$. Then the ideal $Fitt_j(\mathcal{K}) = I_{m-j}(\mathbf{F})$ is called the $j$-th Fitting ideal of $\mathcal{K}$. Here, we make the convention that $Fitt_j(\mathcal{K}) = k[\z]$ for $j \geq m$, and that $Fitt_j(\mathcal{K}) = 0$ for $j < {\rm max}\{m-l,0\}$.
\end{definition}

 We remark that $Fitt_j(\mathcal{K})$ only depends on $\mathcal{K}$. \cite{Cox2005Using} showed that one obtains a presentation matrix of $\mathcal{K}$ by arranging a system of generators of the syzygy module of $\mathcal{K}$ as rows. Let $\mathbf{H}\in k[\z]^{m\times t}$ be composed of a system of generators of $\mathcal{K}$, then the syzygy module of $\mathcal{K}$ is defined as follows: ${\rm Syz}(\mathcal{K}) = \{ \vec{u}\in k[\z]^{1\times m} \mid \vec{u}\mathbf{H} = \vec{0} \}$.

 Let $\mathbf{F}\in k[\z]^{l\times m}$ with rank $r$, and $J$ be a nonzero ideal of $k[\z]$, where $1\leq r \leq l$. We use $\rho(\F)$ to denote the submodule of $k[\z]^{1\times m}$ generated by the rows of $\F$. \cite{Mingsheng2007On} and \cite{Guan2018} proved that the rank of $\rho(\F):J$ is $r$. Let $\mathbf{F}_1\in k[\z]^{s\times m}$ be composed of a system of generators of $\rho(\F):J$, and $\mathbf{F}_2\in k[\z]^{t\times s}$ be composed of a system of generators of ${\rm Syz}(\F_1)$, where $s\geq r$ and $t\geq s-r$. Then, $\mathbf{F}_2$ is a presentation matrix of $\rho(\F):J$. Moreover, $\rho(\F):J$ is a free module of rank $r$ if and only if $Fitt_r(\rho(\F_1))$ generates $k[\z]$, that is, $I_{s-r}(\mathbf{F}_2) = k[\z]$. We refer to \cite{Eisenbud2013} for more details.

 \cite{Mingsheng2005On} proposed a necessary and sufficient condition for MLP factorizations of multivariate polynomial matrices with full row rank.

 \begin{lemma}\label{minor-theorem-Wang}
   Let $\mathbf{F}\in k[\z]^{l\times m}$ be of full row rank. Then the following are equivalent:
   \begin{enumerate}
     \item $\F$ has an MLP factorization;

     \item $\rho(\F):d_l(\F)$ is a free module of rank $l$.
   \end{enumerate}
 \end{lemma}

 \cite{Guan2019} generalized Lemma \ref{minor-theorem-Wang} to the case of multivariate polynomial matrices without full row rank.

 \begin{lemma}\label{minor-theorem-Guan}
   Let $\mathbf{F}\in k[\z]^{l\times m}$ with rank $r$, where $1\leq r \leq l$. Then the following are equivalent:
   \begin{enumerate}
     \item $\F$ has an MLP factorization;

     \item $\rho(\F):I_r(\F)$ is a free module of rank $r$.
   \end{enumerate}
 \end{lemma}

\begin{remark}
 Although Lemma \ref{minor-theorem-Guan} is different from Lemma \ref{minor-theorem-Wang} for the case of $r=l$, Guan et al. have proved that $\rho(\F):I_l(\F) = \rho(\F):d_l(\F)$.
\end{remark}

 Let $a_1,\ldots, a_\beta\in k[\z]$ be all the $r\times r$ minors of $\F$, then $I_r(\F) = \langle a_1,\ldots, a_\beta \rangle$, where $\beta = \binom{l}{r} \cdot \binom{m}{r}$. From Equation (\ref{quotient-module-1}) we have
  \begin{equation}\label{Guan-module-equ-1}
   \rho(\F) : I_r(\F) = (\rho(\F) : a_1) \cap \cdots \cap (\rho(\F) : a_\beta).
 \end{equation}
 When we verify whether $\rho(\F):I_r(\F)$ is a free module of rank $r$, we need to do the following calculation. First, we compute a \gr basis $\{\bar{a}_1,\ldots, \bar{a}_\gamma\}$ of $I_r(\F)$, where $\gamma \leq \beta$. Then,
 \begin{equation}\label{Guan-module-equ-2}
   \rho(\F) : I_r(\F) = (\rho(\F) : \bar{a}_1) \cap \cdots \cap (\rho(\F) : \bar{a}_\gamma).
 \end{equation}
 Second, we obtain a system $\mathcal{G}_i$ of generators of $\rho(\F) : \bar{a}_i$ by computing a \gr basis of a corresponding module (we refer to Section \ref{sec_AE} for more details), where $i =1,\ldots,\gamma$. Third, we compute a \gr basis $\mathcal{G}$ of $\mathcal{G}_1\cap \cdots \cap \mathcal{G}_\gamma$. Finally, we compute a \gr basis of the $r$-th Fitting ideal of the module generated by the elements in $\mathcal{G}$.

 As we all know, the method of computing a \gr basis of the intersection of modules is to introduce new variables. Given that the complexity of \gr basis computations is heavily influenced by the number of variables and the total degrees of polynomials \citep{mayr82,mora1984}, it can be seen that the calculation amount of $\rho(\F):I_r(\F)$ is very large. Therefore, we consider the following problem.

 \begin{problem}\label{main-problem-1}
  Is there a simpler condition that can replace $\rho(\F) : I_r(\F)$ in Lemma \ref{minor-theorem-Guan}?
 \end{problem}

\section{Main Result}\label{sec_MR}

 Let $\mathbf{F}\in k[\z]^{l\times m}$ with rank $r$, where $1\leq r \leq l$. We use Lemma \ref{Youla-minor} to establish a relationship between $\F$ and an arbitrary full row rank submatrix of $\F$, and then solve Problem \ref{main-problem-1}.

 \begin{theorem}\label{main-theorem}
   Let $\mathbf{F}\in k[\z]^{l\times m}$ with rank $r$, and $\F_1 \in k[\z]^{r\times m}$ be an arbitrary full row rank submatrix of $\F$, where $1\leq r \leq l$. Then the following are equivalent:
   \begin{enumerate}
     \item $\F$ has an MLP factorization;

     \item $\rho(\F_1):d_r(\F_1)$ is a free module of rank $r$.
   \end{enumerate}
 \end{theorem}

 \begin{proof}
   $1\rightarrow 2$. Suppose $\F$ has an MLP factorization. Then there exist $\mathbf{G}_0\in k[\z]^{l\times r}$ and $\mathbf{F}_0\in k[\z]^{r\times m}$ such that $\mathbf{F} = \mathbf{G}_0\mathbf{F}_0$ with $\mathbf{F}_0$ being an MLP matrix. Without loss of generality, we assume that the first $r$ rows of $\mathbf{F}$ are $k[\z]$-linearly independent. Let $\F_1\in k[\z]^{r\times m}$ be composed of the first $r$ rows of $\mathbf{F}$, then
   \begin{equation}\label{Conjecture-equ-1}
     \F = \begin{bmatrix}\F_1 \\ \mathbf{C}\end{bmatrix}
        = \begin{bmatrix}\G_{01} \\ \G_{02}\end{bmatrix}\F_0,
   \end{equation}
   where $\G_{01} \in k[\z]^{r\times r}$ is the first $r$ rows of $\G_0$. From Equation (\ref{Conjecture-equ-1}) we have
   \begin{equation}\label{Conjecture-equ-2}
     \F_1 =  \G_{01}\F_0.
   \end{equation}
   According to Lemma \ref{minor-theorem-Wang}, $\rho(\F_1):d_r(\F_1)$ is a free module of rank $r$.

   $2\rightarrow 1$. Assume that $\rho(\F_1):d_r(\F_1)$ is a free module of rank $r$. Using Lemma \ref{minor-theorem-Wang}, there exist $\mathbf{G}_{11}\in k[\z]^{r\times r}$ and $\mathbf{F}_{11}\in k[\z]^{r\times m}$ such that $\mathbf{F}_1 = \mathbf{G}_{11}\mathbf{F}_{11}$ with $\mathbf{F}_{11}$ being an MLP matrix. Since $\F_1$ is an arbitrary $r\times m$ submatrix of $\F$, there exists an elementary transformation matrix $\mathbf{U}\in k^{l\times l}$ such that $\F_1$ is the first $r$ rows of $\bar{\F}$, where $\bar{\F} = \mathbf{U} \F$. Let $\bar{\F} = [\F_1^{\rm T} ~ \mathbf{C}^{\rm T}]^{\rm T}$, where $\mathbf{C}\in k[\z]^{(l-r)\times m}$ is the last $(l-r)$ rows of $\bar{\F}$. Then,
   \begin{equation}\label{Conjecture-equ-3}
     \bar{\F} = \mathbf{U} \F = \begin{bmatrix}\F_1 \\ \mathbf{C}\end{bmatrix}
        = \begin{bmatrix}\G_{11}\F_{11} \\ \mathbf{C}\end{bmatrix}
        = \begin{bmatrix}\G_{11} & \mathbf{0}_{r\times (l-r)} \\ \mathbf{0}_{(l-r)\times r}  & \mathbf{I}_{(l-r)\times (l-r)}\end{bmatrix}\begin{bmatrix}\F_{11} \\ \mathbf{C}\end{bmatrix}.
   \end{equation}
   Because $\F_{11}\in k[\z]^{r\times m}$ is a full row rank matrix, there exists another elementary transformation matrix $\mathbf{V}\in k^{m\times m}$ such that the first $r$ columns of $\bar{\F}_{11}$ are $k[\z]$-linearly independent, where $\bar{\F}_{11} = \F_{11}\mathbf{V}$. It follows from ${\rm det}(\mathbf{V}) = 1$ that $\bar{\F}_{11}\mathbf{V}^{-1} = \F_{11}$. According to the Binet-Cauchy formula, we obtain $d_r(\bar{\F}_{11}) \mid d_r(\F_{11})$. This implies that $d_r(\bar{\F}_{11})$ is a nonzero constant. Therefore, $\bar{\F}_{11}$ is an MLP matrix. Suppose that
   \begin{equation}\label{Conjecture-equ-4}
          \begin{bmatrix}\F_{11} \\ \mathbf{C}\end{bmatrix}\mathbf{V}
        = \begin{bmatrix}\mathbf{A}_{11} & \mathbf{A}_{12} \\ \mathbf{A}_{21}  & \mathbf{A}_{22}\end{bmatrix},
   \end{equation}
   where $\mathbf{A}_{11}\in k[\z]^{r\times r}$, $\mathbf{A}_{12}\in k[\z]^{r\times (m-r)}$, $\mathbf{A}_{21}\in k[\z]^{(l-r)\times r}$, and $\mathbf{A}_{22}\in k[\z]^{(l-r)\times (m-r)}$. Then, ${\rm det}(\mathbf{A}_{11}) \neq 0$ and $[\mathbf{A}_{11} ~ \mathbf{A}_{12}]$ is an MLP matrix. By Lemma \ref{Youla-minor}, we get
   \begin{equation}\label{Conjecture-equ-5}
     \begin{bmatrix}\F_{11} \\ \mathbf{C}\end{bmatrix}\mathbf{V}
     = \begin{bmatrix}\mathbf{I}_{r\times r} \\ \mathbf{A}_{21}\mathbf{A}_{11}^{-1} \end{bmatrix}\begin{bmatrix}\mathbf{A}_{11} & \mathbf{A}_{12} \end{bmatrix}
     = \begin{bmatrix}\mathbf{I}_{r\times r} \\ \mathbf{A}_{21}\mathbf{A}_{11}^{-1} \end{bmatrix} \bar{\F}_{11}.
   \end{equation}
   Combining Equation (\ref{Conjecture-equ-3}) and Equation (\ref{Conjecture-equ-5}), we have
   \begin{equation}\label{Conjecture-equ-6}
     \mathbf{U}\F \mathbf{V}=
        \begin{bmatrix}\G_{11} & \mathbf{0}_{r\times (l-r)} \\ \mathbf{0}_{(l-r)\times r}  & \mathbf{I}_{(l-r)\times (l-r)}\end{bmatrix}
        \begin{bmatrix}\mathbf{I}_{r\times r} \\ \mathbf{A}_{21}\mathbf{A}_{11}^{-1} \end{bmatrix} \bar{\F}_{11}
        =\begin{bmatrix}\G_{11} \\ \mathbf{A}_{21}\mathbf{A}_{11}^{-1} \end{bmatrix} \bar{\F}_{11}.
   \end{equation}
   As $\mathbf{U}$ and $\mathbf{V}$ are two elementary transformation matrices, from Equation (\ref{Conjecture-equ-6}) we can derive
   \begin{equation}\label{Conjecture-equ-7}
     \F =\mathbf{U}^{-1}\begin{bmatrix}\G_{11} \\ \mathbf{A}_{21}\mathbf{A}_{11}^{-1} \end{bmatrix} \bar{\F}_{11}\mathbf{V}^{-1}
        =\mathbf{U}^{-1}\begin{bmatrix}\G_{11} \\ \mathbf{A}_{21}\mathbf{A}_{11}^{-1} \end{bmatrix}\F_{11}.
   \end{equation}
   Let $\G_0 = \mathbf{U}^{-1}\begin{bmatrix}\G_{11} \\ \mathbf{A}_{21}\mathbf{A}_{11}^{-1} \end{bmatrix}$ and $\F_0 = \F_{11}$, then $\F = \G_0\F_0$. Thus, $\F$ has an MLP factorization, and the proof is completed.
 \end{proof}

 \begin{remark}
  Theorem \ref{main-theorem} is the same as Lemma \ref{minor-theorem-Wang} for the case of $r =l$.
 \end{remark}

 According to the proof process of sufficiency in Theorem \ref{main-theorem}, we can propose a new constructive algorithm to compute an MLP factorization of $\mathbf{F}$. We will introduce the new algorithm in detail in the following section.

\section{Algorithm and Examples}\label{sec_AE}

 Let $\mathbf{F}\in k[\z]^{l\times m}$ with rank $r$, and $\F_1 \in k[\z]^{r\times m}$ be an arbitrary full row rank submatrix of $\F$, where $1\leq r \leq l$. Suppose $\rho(\F_1):d_r(\F_1)$ is a free module of rank $r$, then $\F$ has an MLP factorization. Now, we need to design an algorithm to compute $\G_0\in k[\z]^{l\times r}$ and $\F_0\in k[\z]^{r\times m}$ such that $\F = \G_0\F_0$ with $\F_0$ being an MLP matrix.

 Computing free bases of free modules is a crucial step in the process of matrix factorizations. \cite{Fabianska2007Applications} first designed a Maple package, which is called QUILLENSUSLIN, to compute free bases of free modules. Based on this fact, we will implement our algorithm on Maple.

 We have two problems to solve. The first one is how to compute a system of generators of $\rho(\F_1):d_r(\F_1)$, and another one is how to compute $\G_{11}\in k[\z]^{r\times r}$ such that $\F_1 = \G_{11} \F_{11}$, where $\F_{11}$ is composed of a free basis of $\rho(\F_1):d_r(\F_1)$. We can use the commands ``quotient" and ``lift" on the computer algebra system Singular \citep{DGPS2016} to solve the two problems. However, we need to solve these problems on Maple.

 \cite{Mingsheng2005On} proved that there are one to one correspondences between the two modules: $\rho(\F_1):d_r(\F_1)$ and ${\rm Syz}([\F_1^{\rm T} ~~ -d_r(\F_1)\cdot \mathbf{I}_{m\times m}]^{\rm T})$. That is, we compute a \gr basis $\{ [\vec{g}_1,\vec{f}_1],\ldots,[\vec{g}_s,\vec{f}_s]\}$ of ${\rm Syz}([\F_1^{\rm T} ~~ -d_r(\F_1)\cdot \mathbf{I}_{m\times m}]^{\rm T})$, then $\{\vec{f}_1,\ldots,\vec{f}_s\}$ is a system of generators of $\rho(\F_1):d_r(\F_1)$, where $[\vec{g}_i,\vec{f}_i]\in k[\z]^{1\times (r+m)}$ and $i=1,\ldots,s$.

 Now, we solve the second problem. Let $\F_1$ be composed of $\{\vec{f}_1,\ldots,\vec{f}_r\}$ and $\F_{11}$ be composed of $\{\vec{h}_1,\ldots,\vec{h}_r\}$, where $\vec{f}_i,\vec{h}_j\in k[\z]^{1\times m}$ and $1\leq i,j \leq r$. It follows from $\rho(\F_1) \subset \rho(\F_{11})$ that $\vec{f}_i \in \langle \vec{h}_1,\ldots,\vec{h}_r \rangle$ for each $i$. According to the division algorithm in $k[\z]^{1\times m}$ \citep{Cox2005Using}, we use $\{\vec{h}_1,\ldots,\vec{h}_r\}$ to reduce $\vec{f}_i$ and obtain the following equation:
 \begin{equation}\label{DFT-equ-1}
   \vec{f}_i = a_{i1}\vec{h}_1 + \cdots + a_{ir}\vec{h}_r + \vec{v}_i, ~ i=1,\ldots,r,
 \end{equation}
 where $a_{ij}\in k[\z]$ and $\vec{v}_i\in k[\z]^{1\times m}$. However, $\vec{v}_i$ may be a nonzero vector since $\{\vec{h}_1,\ldots,\vec{h}_r\}$ is not a \gr basis. Hence, we first need to compute a \gr basis $\{\vec{g}_1,\ldots,\vec{g}_s\}$ of $\langle \vec{h}_1,\ldots,\vec{h}_r \rangle$, where $s\geq r$. In the calculation process, we record the relationship between $\{\vec{h}_1,\ldots,\vec{h}_r\}$ and $\{\vec{g}_1,\ldots,\vec{g}_s\}$. That is,
 \begin{equation}\label{DFT-equ-2}
   \vec{g}_i = p_{i1}\vec{h}_1 + \cdots + p_{ir}\vec{h}_r, ~ i=1,\ldots,s.
 \end{equation}
 Then, we use $\{\vec{g}_1,\ldots,\vec{g}_s\}$ to reduce $\vec{f}_i$ and get
 \begin{equation}\label{DFT-equ-3}
   \vec{f}_i = q_{i1}\vec{g}_1 + \cdots + q_{is}\vec{g}_s, ~ i=1,\ldots,r.
 \end{equation}
 Let $\mathbf{P} =$ \begin{footnotesize}$\begin{bmatrix}
     p_{11}   &   \cdots   & p_{1r}   \\
     \vdots   &   \ddots   & \vdots   \\
     p_{s1}   &   \cdots   & p_{sr}
   \end{bmatrix}$\end{footnotesize} and $\mathbf{Q} =$ \begin{footnotesize}$\begin{bmatrix}
     q_{11}   &   \cdots   & q_{1s}   \\
     \vdots   &   \ddots   & \vdots   \\
     q_{r1}   &   \cdots   & q_{rs}
   \end{bmatrix}$\end{footnotesize}. Combining Equation (\ref{DFT-equ-2}) and Equation (\ref{DFT-equ-3}), we have
 \begin{equation}\label{DFT-equ-4}
  \F_1 = \G_{11}\F_{11} = (\mathbf{Q}\mathbf{P})\F_{11}.
 \end{equation}

 \cite{Lu2020Solving} designed a Maple package, which is called poly-matrix-equation, for solving multivariate polynomial matrix Diophantine equations. We can use this package to implement the above calculation process.

 Now, we can propose a new constructive algorithm to compute MLP factorizations of polynomial matrices without full row rank.

\begin{algorithm}[!ht]
 \DontPrintSemicolon
 \SetAlgoSkip{}
 \LinesNumbered
 \SetKwInOut{Input}{Input}
 \SetKwInOut{Output}{Output}

 \Input{$\mathbf{F}\in k[\z]^{l\times m}$.}

 \Output{an MLP factorization of $\mathbf{F}$.}

 \Begin{

  compute the rank $r$ of $\F$;

  perform elementary row transformations on $\F$, such that the first $r$ rows of $\bar{\F}$ are $k[\z]$-linearly independent, where $\bar{\F} = \mathbf{U}\F$ and $\mathbf{U}\in k^{l\times l}$ is an elementary transformation matrix;

  compute $d_r(\F_1)$, where $\F_1$ is composed of the first $r$ rows of $\bar{\F}$;

  compute a \gr basis $\{ [\vec{g}_1,\vec{f}_1],\ldots, [\vec{g}_s,\vec{f}_s]\}$ of ${\rm Syz}([\F_1^{\rm T} ~ -d_r(\F_1)\cdot \mathbf{I}_{m\times m}]^{\rm T})$;

  compute a \gr basis $\{\vec{h}_1,\ldots,\vec{h}_t\}$ of ${\rm Syz}(\F'_1)$ and use it to constitute $\mathbf{H}\in k[\z]^{t\times s}$, where $\F'_1\in k[\z]^{s\times m}$ is composed of $\{\vec{f}_1,\ldots,\vec{f}_s\}$;

  compute a \gr basis $\mathcal{G}$ of $I_{s-r}(\mathbf{H})$;

  \If{$\mathcal{G} \neq \{1\}$}
  {
    {\bf return} $\F$ has no MLP factorizations.
  }

  use the QUILLENSUSLIN package to compute a free basis of $\rho(\F'_1)$ and use it to make up $\mathbf{F}_{11}\in k[\z]^{r\times m}$;

  use the poly-matrix-equation package to compute $\G_{11}\in k[\z]^{r\times r}$ such that $\F_1 = \G_{11}\F_{11}$;

  perform elementary column transformations on $\F_{11}$, such that the first $r$ columns of $\bar{\F}_{11}$ are $k[\z]$-linearly independent, where $\bar{\F}_{11} = \F_{11}\mathbf{V}$ and $\mathbf{V}\in k^{m\times m}$ is an elementary transformation matrix;

  compute $\mathbf{C}\mathbf{V}$, where $\mathbf{C}\in k[\z]^{(l-r)\times m}$ is the last $(l-r)$ rows of $\bar{\F}$;

  compute $\mathbf{A}_{21}\mathbf{A}_{11}^{-1}$, where $\mathbf{A}_{11}$ is composed of the first $r$ columns of $\bar{\F}_{11}$, and $\mathbf{A}_{21}$ is composed of the first $r$ columns of $\mathbf{C}\mathbf{V}$;

  {\bf return} ($\mathbf{U}^{-1}\begin{bmatrix}\G_{11} \\ \mathbf{A}_{21}\mathbf{A}_{11}^{-1}\end{bmatrix}, \F_{11}$).
 }
 \caption{MLP factorizations}
 \label{MLP_Algorithm}
 \end{algorithm}

 From Algorithm \ref{MLP_Algorithm} we have $\rho(\F'_1) = \rho(\F_1):d_r(\F_1)$ in step 6 and $I_{s-r}(\mathbf{H}) = Fitt_r(\rho(\F'_1))$ in step 7. Moreover, $\mathcal{G} \neq \{1\}$ in step 8 implies that $\rho(\F_1):d_r(\F_1)$ is not a free module of rank $r$. If $s=r$ in step 6, then $\F'_1$ is a full row rank matrix. It follows that $\rho(\F_1):d_r(\F_1)$ is a free module of rank $r$ and the rows of $\F'_1$ constitute a free basis of $\rho(\F_1):d_r(\F_1)$. In this case, we do not need to compute a \gr basis of ${\rm Syz}(\F'_1)$ and perform the calculation from step 11.

 We use the two examples in \cite{Guan2019} to illustrate the calculation process of Algorithm \ref{MLP_Algorithm}.

 \begin{example}\label{example-1}
  {\rm Let
  \[\mathbf{F} =
  \begin{bmatrix}
       z_1^2z_2+z_1^2   & z_1  & 0    \\
       z_1z_3^2-z_1z_3  & 0  & z_2z_3-z_2+z_3-1 \\
       2z_1^2z_2z_3-z_1^2z_2+z_1^2z_3-z_1^2    &  z_1z_3-z_1  & z_1z_2^2+z_1z_2
   \end{bmatrix}\]
  be a multivariate polynomial matrix in $\mathbb{C}[z_1,z_2,z_3]^{3\times 3}$, where $z_1>z_2>z_3$ and $\mathbb{C}$ is the complex field.

  It is easy to compute that the rank of $\F$ is $2$, and the first $2$ rows of $\F$ are $\mathbb{C}[z_1,z_2,z_3]$-linearly independent. Let $\F_1 \in \mathbb{C}[z_1,z_2,z_3]^{2\times 3}$ be composed of the first $2$ rows of $\F$, then $d_2(\F_1) = z_1z_3-z_1$. We compute a \gr basis of ${\rm Syz}([\F_1^{\rm T} ~ -d_2(\F_1)\cdot \mathbf{I}_{3\times 3}]^{\rm T})$ and obtain
  \[\{[0, ~ z_1, ~ z_1z_3, ~ 0, ~ z_2+1],~[z_3-1, ~ 0, ~ z_1z_2+z_1, ~ 1, ~ 0]\}.\]
  Now, we get a system of generators of $\rho(\F_1):d_2(\F_1)$ as follows
  \[\{[z_1z_3, ~ 0, ~ z_2+1],~[z_1z_2+z_1, ~ 1, ~ 0]\}.\]
  Let
  \[\mathbf{F}_1' =
  \begin{bmatrix}
   z_1z_3   &      0     &   z_2+1       \\
   z_1z_2+z_1 & 1 & 0
   \end{bmatrix}.\]
  Since ${\rm rank}(\mathbf{F}_1') = 2$, $\mathbf{F}_1'$ is a full row rank matrix. Then, $\rho(\F_1):d_2(\F_1)$ is a free module of rank $2$, and the rows of $\mathbf{F}_1'$ constitute a free basis of $\rho(\F_1):d_2(\F_1)$. Let $\F_{11} = \F_1'$, we use the poly-matrix-equation package to compute $\G_{11}\in \mathbb{C}[z_1,z_2,z_3]^{2\times 2}$ such that $\F_1 = \G_{11}\F_{11}$, and obtain
  \[\mathbf{G}_{11} =
  \begin{bmatrix}
   0   &     z_1         \\
   z_3-1 &   0
   \end{bmatrix}.\]
  Note that the first $2$ columns of $\mathbf{F}_{11}$ are $\mathbb{C}[z_1,z_2,z_3]$-linearly independent. Let
  \[\mathbf{A}_{11}=
    \begin{bmatrix}
     z_1z_3   &      0  \\
     z_1z_2+z_1 & 1
   \end{bmatrix} \text{ and }
   \mathbf{A}_{21} = \begin{bmatrix}
    2z_1^2z_2z_3-z_1^2z_2+z_1^2z_3-z_1^2    &  z_1z_3-z_1
   \end{bmatrix},\]
   then
   \[\mathbf{A}_{21}\mathbf{A}_{11}^{-1}=
    \begin{bmatrix}
    z_1z_2   &  z_1z_3-z_1
   \end{bmatrix}.\]
  Therefore, $\F$ has an MLP factorization:
  \[\F = \begin{bmatrix}\mathbf{G}_{11} \\  \mathbf{A}_{21}\mathbf{A}_{11}^{-1}\end{bmatrix}\F_{11} =
    \begin{bmatrix}
     0   &     z_1         \\
     z_3-1 &   0    \\
     z_1z_2   &  z_1z_3-z_1
   \end{bmatrix}
   \begin{bmatrix}
   z_1z_3   &      0     &   z_2+1       \\
   z_1z_2+z_1 & 1 & 0
   \end{bmatrix}.\]}
 \end{example}

 \begin{example}\label{example-2}
  {\rm Let
  \[\mathbf{F} =
  \begin{bmatrix}
       z_1z_2+z_1-z_2-1   & 0  & z_3    \\
       z_2+1  & z_2+1  & z_1-1 \\
       z_1z_2+z_1    &  z_2+1 & z_1+z_3-1
   \end{bmatrix}\]
  be a multivariate polynomial matrix in $\mathbb{C}[z_1,z_2,z_3]^{3\times 3}$, where $z_1>z_2>z_3$ and $\mathbb{C}$ is the complex field.

  It is easy to compute that the rank of $\F$ is $2$, and the first $2$ rows of $\F$ are $\mathbb{C}[z_1,z_2,z_3]$-linearly independent. Let $\F_1 \in \mathbb{C}[z_1,z_2,z_3]^{2\times 3}$ be composed of the first $2$ rows of $\F$, then $d_2(\F_1) = z_2+1$. We compute a \gr basis of ${\rm Syz}([\F_1^{\rm T} ~ -d_2(\F_1)\cdot \mathbf{I}_{3\times 3}]^{\rm T})$ and obtain a system of generators of $\rho(\F_1):d_2(\F_1)$ as follows
  \[\{[z_2+1, ~ z_2+1, ~ z_1-1],~[z_1z_2+z_1-z_2-1, ~ 0, ~ z_3],~[z_1^2-2z_1-z_3+1, ~ -z_3, ~ 0]\}.\]
  Let
  \[\mathbf{F}_1' =
  \begin{bmatrix}
   z_2+1 & z_2+1 & z_1-1     \\
   z_1z_2+z_1-z_2-1 & 0 & z_3 \\
   z_1^2-2z_1-z_3+1 & -z_3 & 0
   \end{bmatrix},\]
  then a \gr basis of ${\rm Syz}(\F_1')$ is $\{[-z_3, ~ z_1-1, ~ -z_2-1]\}$. Let $\mathbf{H} =[-z_3 ~~~ z_1-1 ~~~ -z_2-1]$, then
  \[ Fitt_2(\rho(\mathbf{F}_1')) = I_1(\mathbf{H}) \neq \mathbb{C}[z_1,z_2,z_3].\]
  This implies that $\rho(\F_1):d_2(\F_1)$ is not a free module of rank $2$. Then, $\F$ has no MLP factorizations.}
 \end{example}

\section{Comparative Performance}\label{sec_edata}

 The above two examples show that Algorithm \ref{MLP_Algorithm} is simpler than the algorithm, which is called Guan's algorithm, proposed by \cite{Guan2019}. To illustrate the advantages of our algorithm, we first compare the main differences between the two algorithms.

 \begin{table}[H]
 \centering
 \vskip -15pt
 \caption{The comparison of two MLP factorization algorithms}
 \label{table_compare_algorithms}
 \vskip 3pt
 \begin{tabular}{cccccc}
  \Xhline{1.5pt}
    Main step &  Guan's algorithm & Algorithm \ref{MLP_Algorithm}  \\
  \Xhline{1pt}
    1 & $\rho(\F):I_r(\F)$  & $\rho(\F_1):d_r(\F_1)$  \\
    2 & $\F = \G_0\F_0$  &  $\F_1 = \G_{11}\F_{11}$ and $\mathbf{A}_{21}\mathbf{A}_{11}^{-1}$ \\
  \Xhline{1.5pt}
 \end{tabular}
  \vskip -8pt
 \end{table}

 The symbols in the above table are the same as those in Lemma \ref{minor-theorem-Guan} and Theorem \ref{main-theorem}. From Table \ref{table_compare_algorithms}, we can get the following preliminary conclusions: first, the calculation of the main step 1 of Algorithm \ref{MLP_Algorithm} is faster than that of Guan's algorithm in almost all cases; second, although in Algorithm \ref{MLP_Algorithm} we need to compute $\mathbf{A}_{21}\mathbf{A}_{11}^{-1}$ additionally, the scale of equation $\F_1 = \G_{11}\F_{11}$ is smaller than that of equation $\F = \G_0\F_0$.

 Next, we will show from the specific experimental data that Algorithm \ref{MLP_Algorithm} is more efficient than Guan's algorithm. The two algorithms have been implemented by us on the computer algebra system Maple. The implementations of the two algorithms have been tried on a number of examples including the two examples in Section \ref{sec_AE}. Please see the Appendix A for all examples. For interested readers, more comparative examples can be generated by the codes at: \url{http://www.mmrc.iss.ac.cn/~dwang/software.html}.

 \begin{table}[H]
 \centering
 \vskip -15pt
 \caption{Comparative performance of MLP factorization algorithms}
 \label{table_compare_time}
 \vskip 3pt
 \begin{tabular}{cccccc}
  \Xhline{1.5pt}
   Example &  \makecell{Guan's algorithm \\ $t_1$ (sec)} & \makecell{Algorithm \ref{MLP_Algorithm} \\ $t_2$ (sec)} & \makecell{Time comparison \\ $t_1/t_2$} \\
  \Xhline{1pt}
    $\F_1$ & 0.257  & 0.037 &  6.95 \\
    $\F_2$ & 0.263  & 0.044 &  5.98 \\
    $\F_3$ & 0.132  & 0.058 &  2.28 \\
    $\F_4$ & 0.407  & 0.063 &  6.46 \\
    $\F_5$ & 3.060  & 0.151 &  20.26 \\
    $\F_6$ & 4.275  & 0.283 &  15.11 \\
    $\F_7$ & 9.037  & 0.330 &  27.38 \\
    $\F_8$ & 17.306 & 0.549 &  31.52 \\
  \Xhline{1.5pt}
 \end{tabular}
 \vskip -8pt
 \end{table}

 In Table \ref{table_compare_time}, timings were obtained on an Intel(R) Xeon(R) CPU E7-4809 v2 @ 1.90GHz and 756GB of RAM, and each time is an average of 100 repetitions of the corresponding algorithm. As is evident from Table \ref{table_compare_time}, our algorithm performs better than Guan's algorithm, especially when the size of entries in matrices becomes larger and larger.

\section{Concluding Remarks}\label{sec_conclusions}

 We have given a new necessary and sufficient condition for the existence of MLP factorizations of multivariate polynomial matrices in this paper. All cases with matrices being full row rank and  non-full row rank are considered. Based on the new result, a constructive algorithm for computing MLP factorizations has been proposed. We have implemented Algorithm \ref{MLP_Algorithm} and Guan's algorithm on Maple, and the experimental data in Table \ref{table_compare_time} suggests that Algorithm \ref{MLP_Algorithm} is superior in practice in comparison with Guan's algorithm. This is due to the fact that we can determine whether $\F$ has an MLP factorization through less calculations and requires less time to calculate $\G_0$.

\section*{Acknowledgments}

 This research was supported in part by the CAS Key Project QYZDJ-SSW-SYS022.

\bibliographystyle{elsarticle-harv}

\bibliography{MLP_MF}

\appendix
\section{Examples for Table \ref{table_compare_time}}

 For all examples, the monomial orders used on $k[\z]$ and $k[\z]^{1\times m}$ are degree reverse lexicographic order and position over term, respectively. $k$ is the complex field $\mathbb{C}$, and $z_1>\cdots>z_n$.

 \begin{enumerate}
 \item $\F_1\in \mathbb{C}[z_1,z_2,z_3]^{3\times 3}$ is as follows, and it has no MLP factorizations.
   \begin{footnotesize}
    \[\mathbf{F}_1 =
     \begin{bmatrix}
       z_1z_2+z_1-z_2-1 & 0 & z_3   \\
       z_2+1 & z_2+1 & z_1-1 \\
       z_1z_2+z_1 & z_2+1 & z_1+z_3-1
     \end{bmatrix}.\]
    \end{footnotesize}

 \item $\F_2\in \mathbb{C}[z_1,z_2,z_3]^{3\times 3}$ is as follows, and it has no MLP factorizations.
   \begin{footnotesize}
    \[\mathbf{F}_2 =\begin{bmatrix}
       z_1z_2-z_2   & 0  & z_3+1    \\
       0  & z_1z_2-z_2  & z_1^2-2z_1+1 \\
       z_1^2z_2-z_1z_2    &  z_1z_2^2-z_2^2  & z_1^2z_2-2z_1z_2+z_1z_3+z_1+z_2
     \end{bmatrix}.\]
    \end{footnotesize}

 \item $\F_3\in \mathbb{C}[z_1,z_2,z_3]^{3\times 3}$ is as follows, and it has an MLP factorization.
   \begin{footnotesize}
    \[\mathbf{F}_3 =
     \begin{bmatrix}
       z_1z_2^2 & z_1z_3^2 & z_2^2z_3+z_3^3   \\
       z_1z_2 & 0 & z_2z_3 \\
       0 & z_1^2z_3 & z_1z_3^2
     \end{bmatrix}.\]
   \end{footnotesize}

 \item $\F_4\in \mathbb{C}[z_1,z_2,z_3]^{3\times 3}$ is as follows, and it has an MLP factorization.
    \begin{footnotesize}
    \[\mathbf{F}_4 =
      \begin{bmatrix}
       z_1^2z_2+z_1^2 & z_1 & 0  \\
       z_1z_3^2-z_1z_3 & 0 & z_2z_3-z_2+z_3-1  \\
       2z_1^2z_2z_3-z_1^2z_2+z_1^2z_3-z_1^2 & z_1z_3-z_1
       & z_1z_2^2+z_1z_2
      \end{bmatrix}.\]
    \end{footnotesize}

 \item $\F_5\in \mathbb{C}[z_1,z_2,z_3]^{3\times 3}$ is as follows, and it has an MLP factorization.
 \begin{footnotesize}
  \begin{equation*}
   \left\{
    \begin{array}{ll}
     \F_5[1,1]=z_1^2-z_1 , \\
     \F_5[1,2]=-z_2z_3+z_1-z_3 , \\
     \F_5[1,3]=z_1z_3-2z_1-z_3 , \\
     \F_5[2,1]=z_1^3z_2z_3-z_1^3z_3-z_1^2z_2z_3-z_1^2z_2+
               2z_1^2z_3+z_1^2+
               \\~~~~~~~~~~~~~~
               z_1z_2-z_1z_3-z_1 , \\
     \F_5[2,2]=-z_1^2z_2^2z_3-z_1z_2^2z_3^2-z_1^2z_2z_3-z_1z_2z_3^2+
               z_1z_2^2-2z_1^2z_3+
               \\~~~~~~~~~~~~~~
               z_2^2z_3-z_2z_3^2+z_1z_2+z_1z_3+z_2z_3-z_3^2+2z_1 ,\\
     \F_5[2,3]=z_1^2z_2z_3^2-3z_1^2z_2z_3-z_1^2z_3^2-
               z_1z_2z_3^2+z_1^2z_3-z_1z_2z_3+
               \\~~~~~~~~~~~~~~
               z_1z_3^2+3z_1z_2-z_1z_3+z_2z_3-z_3^2-z_1 , \\
     \F_5[3,1]=z_1^2z_2^3-z_1^2z_2^2-z_1z_2^3+z_1z_2^2+z_1z_2-z_2 ,\\
     \F_5[3,2]=-z_1z_2^4-z_2^4z_3-z_1z_2^3-z_2^3z_3-
               2z_1z_2^2+z_2^3+2z_2^2+2z_2,\\
     \F_5[3,3]=z_1z_2^3z_3-3z_1z_2^3-z_1z_2^2z_3-z_2^3z_3+
               z_1z_2^2+z_2^2+z_2z_3-z_2.
    \end{array}
   \right.
  \end{equation*}
 \end{footnotesize}

 \item $\F_6\in \mathbb{C}[z_1,z_2,z_3]^{3\times 3}$ is as follows, and it has an MLP factorization.
 \begin{footnotesize}
  \begin{equation*}
   \left\{
    \begin{array}{ll}
     \F_6[1,1]=z_1^3z_2^2+2z_1^3z_2+z_1^3-z_1z_2^2-z_1z_2z_3-z_1z_2-
               z_1z_3+
               \\~~~~~~~~~~~~~~~
               z_2z_3-z_2+z_3-1 , \\
     \F_6[1,2]=-z_1z_2^3z_3+z_1^2z_2^2-3z_1z_2^2z_3-2z_2^3z_3+2z_1^2z_2+
               z_1z_2^2+
               \\~~~~~~~~~~~~~~~
               z_2^3-3z_1z_2z_3-6z_2^2z_3+z_1^2+2z_1z_2+3z_2^2-z_1z_3-
               \\~~~~~~~~~~~~~~~
               7z_2z_3+z_1+4z_2-3z_3+2 , \\
     \F_6[1,3]=z_1^2z_2^2z_3-2z_1^2z_2^2+2z_1^2z_2z_3-4z_1^2z_2-
               2z_1z_2^2+z_1^2z_3-2z_2^2z_3-
               \\~~~~~~~~~~~~~~~
               z_2z_3^2-2z_1^2-4z_1z_2+z_2^2-z_2z_3-z_3^2-2z_1+z_3-1,\\
     \F_6[2,1]=z_1^2z_3-z_1^2+2z_1z_2-z_1z_3-2z_2+1 , \\
     \F_6[2,2]=z_1z_2^2+2z_2^3-z_2z_3^2+2z_1z_2+3z_2^2+z_1z_3+2z_2z_3-
               \\~~~~~~~~~~~~~~~
               z_3^2+2z_2+2z_3-2 , \\
     \F_6[2,3]=z_1z_3^2+z_1z_2+2z_2^2-3z_1z_3+2z_2z_3-
               z_3^2+3z_1-3z_2+z_3+1 , \\
     \F_6[3,1]=z_1^2z_2z_3+z_1^2z_2+2z_1^2z_3-z_1z_2z_3+
               2z_1^2-z_1z_2-2z_1z_3-2z_1 , \\
     \F_6[3,2]=-z_2^2z_3^2+z_1z_2z_3-z_2^2z_3-3z_2z_3^2+
               z_1z_2+2z_1z_3-
               \\~~~~~~~~~~~~~~~
               3z_2z_3-2z_3^2+2z_1-2z_3 , \\
     \F_6[3,3]=z_1z_2z_3^2-z_1z_2z_3+2z_1z_3^2-z_2z_3^2-2z_1z_2-2z_1z_3-
               \\~~~~~~~~~~~~~~~
               z_2z_3-2z_3^2-4z_1-2z_3 .
    \end{array}
   \right.
  \end{equation*}
 \end{footnotesize}

 \item $\F_7\in \mathbb{C}[z_1,z_2,z_3]^{3\times 4}$ is as follows, and it has an MLP factorization.
 \begin{footnotesize}
  \begin{equation*}
   \left\{
    \begin{array}{ll}
     \F_7[1,1]=2z_1^3z_2^3z_3^2-z_1^3z_2^3z_3+2z_1^4z_3^3-z_1^4z_3^2+
               z_1^3z_2z_3^2+z_1^3z_3^3+
               \\~~~~~~~~~~~~~~
               3z_1z_2^3z_3+ 2z_1^3z_3^2+z_1z_2^2z_3^2+2z_1z_2^2z_3+
               4z_1^2z_3^2-
               \\~~~~~~~~~~~~~~
               z_1^2z_3+z_1z_2z_3+z_1z_3^2+2z_1z_3+2z_3  , \\
     \F_7[1,2]=2z_1z_2^3z_3^2-z_1z_2^3z_3+2z_1^2z_3^3-
               z_1^2z_3^2+2z_3^2-z_3 , \\
     \F_7[1,3]=z_1^3z_3^2+z_1z_2^2z_3+z_1z_3 , \\
     \F_7[1,4]=2z_1z_2^3z_3^2+z_1^3z_3^3-z_1z_2^3z_3+z_1z_2^2z_3^2+
               2z_1^2z_3^3-
               \\~~~~~~~~~~~~~~
               z_1^2z_3^2+z_1z_3^2+2z_3^2-z_3 , \\
     \F_7[2,1]=-2z_1^2z_3^2+z_1^2z_3-z_1z_2z_3-z_1z_3^2-
               2z_1z_3-2z_3,\\
     \F_7[2,2]=-2z_3^2+z_3 , \\
     \F_7[2,3]=-z_1z_3 , \\
     \F_7[2,4]=-z_1z_3^2-2z_3^2+z_3 , \\
     \F_7[3,1]=-2z_1^4z_2^3z_3+z_1^4z_2^3-2z_1^5z_3^2+z_1^5z_3-
               z_1^4z_2z_3-z_1^4z_3^2-
               \\~~~~~~~~~~~~~~
               3z_1^2z_2^3-2z_1^4z_3-z_1^2z_2^2z_3-2z_1^2z_2^2-
               4z_1^3z_3-
               2z_1^2z_2z_3+
               \\~~~~~~~~~~~~~~
               z_1^3-z_1^2z_3-2z_1^2-2z_1-3z_2-z_3-2 , \\
     \F_7[3,2]=-2z_1^2z_2^3z_3+z_1^2z_2^3-2z_1^3z_3^2+z_1^3z_3-
               \\~~~~~~~~~~~~~~
               2z_1z_3-2z_2z_3+z_1+z_2 , \\
     \F_7[3,3]=-z_1^4z_3-z_1^2z_2^2-z_1^2-1 , \\
     \F_7[3,4]=-2z_1^2z_2^3z_3-z_1^4z_3^2+z_1^2z_2^3-z_1^2z_2^2z_3-
               2z_1^3z_3^2+z_1^3z_3-
               \\~~~~~~~~~~~~~~
               z_1^2z_3-2z_1z_3-2z_2z_3+z_1+z_2-z_3.
    \end{array}
   \right.
  \end{equation*}
 \end{footnotesize}

 \item $\F_8\in \mathbb{C}[z_1,z_2,z_3]^{3\times 3}$ is as follows, and it has an MLP factorization.
 \begin{footnotesize}
  \begin{equation*}
   \left\{
    \begin{array}{ll}
     \F_8[1,1]=z_1^3z_3^2-z_1^3z_3+z_1^2z_2z_3-z_1^2z_3^2+z_2z_3^3+
               z_3^4-z_1^2z_2-
               \\~~~~~~~~~~~~~~
               z_1z_2z_3-z_2z_3^2-2z_3^3+z_1^2+z_1z_2+z_1z_3+z_2z_3+
               \\~~~~~~~~~~~~~~
               2z_3^2-z_1-z_2-2z_3+1,\\
     \F_8[1,2]=z_1^4z_2z_3+z_1^3z_2^2-z_1^3z_2z_3+z_1z_2^2z_3^2+
               z_1z_2z_3^3-z_1^3z_2-z_1^2z_2^2+
               \\~~~~~~~~~~~~~~
               z_1^3z_3-z_1z_2z_3^2+2z_1^2z_2+z_1z_2^2-z_1^2z_3+
               z_1z_2z_3+z_2z_3^2+
               \\~~~~~~~~~~~~~~
               z_3^3-z_1^2-2z_1z_2-z_3^2+z_1+z_2+z_3-1 , \\
     \F_8[1,3]=2z_1^3z_2-z_1^3-2z_1^2z_2+z_2^2z_3+z_2z_3^2+z_1^2+
               \\~~~~~~~~~~~~~~
               z_1z_2+z_1z_3-z_2z_3-z_1 , \\
     \F_8[2,1]=-z_1z_2z_3^3+z_1^2z_3^2+z_1z_2z_3^2-z_2^2z_3^2-
               z_1^2z_3+z_1z_2z_3+z_2^2z_3-
               \\~~~~~~~~~~~~~~
               z_3^3-z_1z_2+2z_3^2-2z_3+1 , \\
     \F_8[2,2]=-z_1^2z_2^2z_3^2+z_1^3z_2z_3-z_1z_2^3z_3+z_1^2z_2^2-
               2z_1z_2z_3^2+z_1^2z_3+
               \\~~~~~~~~~~~~~~
               z_1z_2z_3-z_2^2z_3-z_3^2+z_3-1 ,\\
     \F_8[2,3]=-2z_1z_2^2z_3+2z_1^2z_2+z_1z_3-z_2z_3-z_1 , \\
     \F_8[3,1]=z_1z_2^2z_3^2+z_1z_3^4-z_1z_2^2z_3+z_2^3z_3+
               2z_1z_2z_3^2-3z_1z_3^3-
               \\~~~~~~~~~~~~~~
               z_3^4-z_2^3-2z_1z_2z_3+z_2^2z_3+5z_1z_3^2+z_3^3-z_2^2-
               \\~~~~~~~~~~~~~~
               5z_1z_3-z_3^2+2z_1-z_3+2 , \\
     \F_8[3,2]=z_1^2z_2^3z_3+z_1^2z_2z_3^3+z_1z_2^4+2z_1^2z_2^2z_3-
               2z_1^2z_2z_3^2-z_1z_2z_3^3+
               \\~~~~~~~~~~~~~~
               z_1z_2^3+3z_1^2z_2z_3+z_1z_2^2z_3+z_1z_3^3-z_2^3+
               z_1z_2z_3-2z_1z_3^2-
               \\~~~~~~~~~~~~~~
               2z_1^2z_2+z_3^3-2z_1z_2+z_2^2+3z_1z_3-2z_1-z_3-2 , \\
     \F_8[3,3]=2z_1z_2^3+z_1z_2z_3^2+3z_1z_2^2+z_1^2z_3-2z_1z_2z_3-
               \\~~~~~~~~~~~~~~
               z_2z_3^2-2z_1^2+2z_1z_2-z_1z_3-2z_1 .
    \end{array}
   \right.
  \end{equation*}
 \end{footnotesize}
 \end{enumerate}

\end{document}